\documentclass[11pt]{article}
\usepackage{amssymb,amsmath}
\usepackage{epsfig}
\usepackage{mathpazo}
\usepackage{graphicx}
\usepackage{textcomp}
\usepackage{hyperref}
\usepackage{amsthm}
\usepackage{mdframed}

\newtheorem{theorem}{Theorem}
\newtheorem*{theorem*}{Theorem}
\newtheorem*{definition*}{Definition}
\newtheorem{proposition}{Proposition}

\newtheorem{definition}{Definition}
\newtheorem{algorithm}{Algorithm}
\newtheorem{claim}{Claim}
\newtheorem{lemma}{Lemma}

\newtheorem{fact}{Fact}

\newtheorem{question}{Question}

\def\Z{{\mathbb{Z}}}
\def\R{\mathbb{R}}

\def\C{\mathbb{C}}

\def\mod{\mbox{mod}}

\def\poly{{\rm poly}}
\def\log{{\rm log}}

\def\Span{{\rm Span}}

\def\modn{(\mod \ N)}

\newcommand{\be}{\begin{eqnarray}}
\newcommand{\ee}{\end{eqnarray}}

\newcommand\round[1]{{\lfloor #1 \rceil}}

\newcommand\ket[1]{{ |{#1} \rangle }}

\def\P{{\sf{P}}}

\def\dist{{\rm dist}}

\newcommand{\eps}{\varepsilon}
\renewcommand{\epsilon}{\varepsilon}

\begin{document}

\title{A Discrete Fourier Transform on Lattices \\ with Quantum Applications} 

\author{Lior Eldar\thanks{Center for Theoretical physics, MIT} \and
Peter W. Shor\thanks{Department of Mathematics and Center for Theoretical physics, MIT}
}

\maketitle

\abstract{

In this work, we introduce a definition of the Discrete Fourier Transform (DFT) on Euclidean lattices in $\R^n$,
that generalizes the $n$-th fold DFT of the integer lattice $\Z^n$ to arbitrary lattices.
This definition is not applicable for every lattice, but can be defined on lattices known as
Systematic Normal Form (SysNF) introduced in \cite{ES16}.
Systematic Normal Form lattices are sets of integer vectors that satisfy a single homogeneous
modular equation, which itself satisfies a certain number-theoretic property.
Such lattices form a dense set in the space of $n$-dimensional lattices, 
and can be used to approximate efficiently any lattice.
This implies that for every lattice $L$ a DFT can be computed efficiently on a lattice near $L$.

Our proof of the statement above uses arguments from quantum computing,
and as an application of our definition we show a quantum algorithm
for sampling from discrete distributions on lattices, that extends our ability to sample efficiently from 
the discrete Gaussian distribution \cite{GPV08} to any distribution that is sufficiently "smooth".
We conjecture that studying the eigenvectors of the newly-defined lattice DFT may provide new 
insights into the structure of lattices, especially regarding hard computational problems, like the shortest vector problem.
}

\section{Introduction}

The Fourier Transform is ubiquitous in the study of lattices in mathematics, 
and in recent years has led to breakthroughs in our understanding of the complexity of lattice problems \cite{AR04,Reg09}.
The Fourier Transform on Euclidean lattices is usually associated with the
Fourier series of lattice-periodic functions:
Let $L\subseteq \R^n$ 
denote some full-rank $n$-dimensional lattice, $L = \Span_{\Z}(B)$, where $B\in GL(n,\R)$.
Consider the set of bounded complex-valued continuous functions $f: \R^n \rightarrow \C$
that are periodic in $L$, i.e.
$$
\forall x\in \R^n, z\in L, \ \ f(x) = f(x + z).
$$
Then the Fourier series of $f$, $\hat{f}: L^* \mapsto \C$, supported on the dual lattice $L^*$ is defined as follows:
$$
\forall z\in L^*, \ \ 
\hat{f}(z) := \frac{1}{\det(B)} \cdot \int_{{\cal P}(L)} f(x) e^{-2\pi i \langle x ,z \rangle} dx,
$$
where ${\cal P}(L)$ is the basic parallelotope of the lattice defined by the image
of $[0,1)^n$ under $B$.
Hence, in this respect, the FT on $n$-dimensional lattices is defined as the $n$-dimensional generalization of the Fourier Series of functions
defined on the unit interval.

The Discrete Fourier Transform (DFT)
of a sequence of $N$ complex numbers $X_0,\hdots, X_{N-1}$ is defined as
$$
\forall k\in \Z_N \quad
\hat{x}_k = \sum_{z=0}^{N-1} X_z e^{-2\pi i x \cdot z / N}.
$$
It is a map between discrete sequences that can be thought of as
a discretization of the Fourier Transform to regularly spaced-grids
in the following sense: the Fourier-Transform of a function $f$ that is periodic on the interval $[0,N] \subseteq \R$,
sampled at integer points $[0,\hdots,N-1]$,
corresponds to the DFT of the sequence derived by sampling $f$ at the points $[0,\hdots,N-1]$.
The DFT has proven to be extremely useful in both engineering and computer science.

Given the interpretation of the DFT as a regularly-spaced sampling of the continuous FT
it is then natural to consider whether one can define the DFT on an arbitrary lattice.
Specifically, it would be desirable to have a definition of the DFT which inherits the
inner-product between lattice vectors.
Such is the case for the trivial lattice $\Z^n$:
for any integer $N$ one can consider the ring of integers modulo $N$, $\Z_N$
and define for any function $f: \Z_N^n \to \C$:
$$
\forall x\in \Z_N^n \ \ 
\hat{f}(x) = \sum_{z\in \Z_N^n} f(z) e^{-2\pi i \langle x,z\rangle/N}.
$$
In this case, the DFT at each point corresponds to sampling
the continuous FT of $f$ at the points of $L = \Z^n$.
Furthermore, this definition corresponds to the Fourier Transform of the 
finite group $\Z_N^n$ with entry-wise addition modulo $N$.

We would like to have this behavior for any arbitrary lattice $L\subseteq \R^n$. 
But to relate to finite groups we need to relate to a finite subset of $L$.
Let $N = \det(L)$.
Then $L$ is periodic in $N$ in each direction, 
i.e. for any $v\in L$ we have $v + N e_i \in L$ for all $i\in [n]$.
Therefore,
it is sufficient to consider the finite lattice $L_N$ as an additive subgroup of the finite vector space
$\Z_N^n$ with addition modulo $N$, instead of $L$ as an additive subgroup of $\R^n$ with real addition.
We define lattice DFT as follows:
\begin{mdframed}
\begin{definition}

\textbf{Lattice DFT}

\noindent
Let $L\subseteq \R^n$ be an $n$-dimensional integer lattice, $N = \det(L)$.
A Discrete Fourier Transform of $L$ (DFT) is 
a Fourier Transform of the finite group $L_N$,
for which the characters $\chi_x(z)$ for $x,z\in L_N$ satisfy:
$$
\forall x,z\in L_N \ \ \chi_x(z) = e^{-2\pi i \langle x,z\rangle/N}.
$$
\end{definition}
\end{mdframed}

\noindent
\\
and so the main question is
\begin{question}
Does there exist a lattice DFT for every lattice?
\end{question}

A natural place to look for a DFT is in the context of finite Abelian groups.
Given a lattice $L$ with determinant $N = \det(L)$, one can restrict his attention to the
set of lattice points with entries in $\Z_N$, and consider this as a finite sub-group $L_N$
of the cube $\Z_N^n$ with entry-wise addition modulo $N$.
Since $L_N$ is a finite Abelian group then by the fundamental theorem of classification of finite Abelian groups
$L_N$ is isomorphic to a product of primary cyclic groups.
Hence, one can define the DFT of $L_N$ by considering the DFT of the individual
prime-power factors $\Z_p^k$ for prime $p$ and integer $k$.
Yet, one can check that generically, the resulting DFT would have an inner-product which is very different from the integer inner-product
modulo $N$
between lattice points.

In this work we answer the question above by showing that one can define the DFT for a certain
dense set of lattices.
Furthermore, we show that this DFT
can be computed efficiently, albeit with a quantum computer.
This dense set of lattices corresponds to lattices of a special form called {\it Systematic Normal Form} (or SysNF for short)
introduced by Eldar and Shor in \cite{ES16}:
\begin{definition}\label{def:SysNF}

\textbf{Systematic Normal Form (SysNF) }\cite{ES16}

\noindent
An integer matrix $B$ is said to be SysNF if $B_{i,i}=1$ for all $i>1$,
$B_{i,j}=0$ for all $i>1, i\neq j$, and $B_{1,1} = N$ satisfies
\be\label{eq:cond1}
\sum_{i>1} B_{1,i}^2 + 1\neq 0\modn.
\ee
\end{definition} 
\noindent
Specifying only the non-zero entries of $B$ - it can be written as:
\be
B
= \left[\begin{array}{ccccc} 
N & b_2 & b_3 & \hdots & b_n \\
 & 1& & &\\
 & & 1 & &  \\
 & & & \ddots & \\
 & & & & 1 
\end{array}
\right]
\ee

These lattices form a dense set in the space of lattices in terms of the Euclidean distance, in the sense
that for every $\eps>0$ and arbitrary lattice $L$, 
there exists an efficiently computable linear map $\sigma$, a large integer $T$, and a SySNF lattice $L'$ such that for every $x\in L$
$\sigma(x)\in L'$ and $\|x - \sigma(x)/T\| \leq \eps \|x\|$.
(See Lemma \ref{lem:SysNF} for a precise statement).

By its definition,  a SysNF lattice is the set of integer vectors that satisfy a certain homogeneous modular equation
(modulo a number $N$) where, in addition, this equation satisfies an extra number-theoretic condition.
Defining lattices as the set of solutions of modular equations is a def-facto standard in the study of lattices (see e.g. \cite{P15}), especially
in the context of random lattices due to Ajtai \cite{Ajt96}.
However, the extra number-theoretic condition in Equation \ref{eq:cond1} wasn't defined prior to \cite{ES16} and, in fact, is used crucially
to establish that such lattices have a DFT.
We discuss this further in sub-section \ref{sec:prior}.

Our proof that DFT can be defined on SysNF lattices is quantum.
Concretely, we provide a quantum circuit implementing the character
map for each lattice point. 
The details of this implementation are given in Section \ref{sec:fourier}.
To do this, we first define a quantum analog of the map above:
\begin{mdframed}
\begin{definition*}

\textbf{Quantum Fourier Transform on SysNF lattices}

\noindent
Let $L\subseteq\R^n$ be a SysNF lattice, $N = \det(L)$.
The Quantum Fourier Transform on $L_N$ is defined for basis states as follows:
\be
\forall x\in L_N, {\cal F}_{L,N}(\ket{x}) = \frac{1}{\sqrt{N^{n-1}}}\sum_{z\in L_N} e^{-2\pi i \langle x,z \rangle/N} \ket{z}.
\ee
\end{definition*}
\end{mdframed}

\noindent
\\
The normalization by $\sqrt{N^{n-1}}$ follows from the fact that there are precisely $N^{n-1}$ points in $L_N$
(see Proposition \ref{prop:1}).
We then show that this map is unitary (and in particular, efficiently computable) thereby establishing 
that the $|L_N|$ characters $\chi_x(z) = e^{-2\pi i \langle x,z\rangle/N}$ for $x\in L_N$ are orthogonal, and hence
form a complete set of inequivalent irreducible representations of $L_N$ - i.e.
a Fourier Transform of the group $L_N$.
\begin{theorem}

\textbf{A Quantum Circuit for lattice DFT}

\noindent
Given is a lattice $L = L(B)$, where $B$ is an $n\times n$ SysNF matrix.
There exists a quantum circuit ${\cal Q}$ of size $\poly(n)$, that implements ${\cal F}_{L,N}$.
In particular, $L$ can be assigned a lattice DFT.
\end{theorem}

As an application of our new definition, 
the above circuit gives rise to an efficient way to sample from any discrete distribution on a lattice,
for sufficiently "nice" functions:
\begin{theorem}(sketch of Theorem \ref{thm:sample})
Let $f$ be a complex-valued function on $\R^n$, and $L\subseteq \R^n$ some lattice, generated by matrix $B$.
Suppose that ${\cal F}$, the FT of $f$, can be generated as a superposition on $\Z^n$
$$
\sum_{x\in \Z^n} f(x) \ket{x}
$$
and ${\cal F}$ is approximately bounded in $\lambda_1(L^*)/ 2^{n/2}$
then one can approximately sample from the following discrete distribution efficiently quantumly:
$$
\forall x\in L \quad \P(x) \propto |f(x)|^2.
$$
\end{theorem}

\subsection{Discussion and Previous Work}\label{sec:prior}

To the best of our knowledge, a Discrete Fourier Transform 
that inherits the Euclidean inner-product and generalizes the DFT of the
integer lattice $\Z^n$ to arbitrary $n$-dimensional lattices has not been defined before.
The standard notion of the Fourier Transform on arbitrary $n$-dimensional lattices relates to the Fourier Series
of lattice-periodic functions, and thus behaves quite differently - and in particular,
is not a map from the lattice onto itself.
Our definition of DFT for lattices cannot be defined for general lattices.
Luckily, however, SysNF lattices form an efficiently computable dense group in the space of lattices, hence
for every lattice, there exists a "nearby" efficiently-computable lattice for which the DFT can be defined.

The Discrete Fourier Transform we define can be viewed as a Fourier Transform
of the discrete group $L_N\subseteq \Z_N^n$ with entry-wise addition modulo $N$,
where the set of irreducible representations used are the $1$-dimensional characters
of the cyclic group of order $N$. 
We note that given any lattice $L$ with $\det(L) = N$ one can define a Fourier Transform
on the finite group $L_N$ using the Fundamental Theorem of Finite Abelian Groups,
but in general this does not give rise to the DFT with the inner-product between lattice
points as in our definition.
Hence, our claim is not that perturbing a lattice to SysNF is necessary to define a finite-group FT,
but rather that perturbing it is sufficient to define a DFT - a FT that inherits the inner-product over integer vectors modulo $N$.
As an added bonus, the DFT on SysNF lattices can be computed on a quantum computer
in time which is polynomial in the dimension of the lattice.

Perturbing lattices to nearby lattices with special structure is not new and has been investigated
by Paz and Schnorr in \cite{Paz87}.
In that reduction, one perturbs a given lattice $L$ to a nearby lattice $L'$ in which
the quotient $\Z^n / L$ is cyclic.
The authors then characterize a lattice $L$ as the set of vectors satisfying a homogeneous modular
equation if and only if the quotient $\Z^n / L$ is cyclic.
Hence the Paz-Schnorr reduction reduces any lattice to the set of solutions of a homogeneous equation modulo
some large integer $N$.
However, the structure of the reduction generates lattices in which $N$ does not generally satisfy
our extra co-primality condition.
Hence the lattices produced by the Paz-Schnorr reduction cannot be assigned a lattice DFT as in our case.

%Even without the reduction of Paz-Schnorr lattices that are the solution sets of modular equations
%are by now de-facto standard in the research of lattices \cite{}, and are considered to be "hard" instances
%of lattice problems, and considering the modulus to be, say a prime number, doesn't change their hardness.
%However, as we shall see, the extra number-theoretic condition $\sum_i B_{1,i}^2 \neq (-1)\modn$ is crucial
%for the DFT to be implemented, and even defined.
%This condition was introduced by Eldar and Shor \cite{ES16} in the definition of SysNF lattices, and
%to the best of our knowledge was not explored before.

In terms of the quantum implementation of the Fourier Transform, we note that effectively,
it is a reduction from the definition of the DFT on $L_N$ to the standard DFT on $\Z^{n-1}$.
That said, it is only because of the extra number-theoretic condition, namely that $\sum_{i>1} B_{1,i}^2 \neq (-1) \modn$
that such a reduction is possible.  This is described in detail in Section \ref{sec:fourier}.
The quantum implementation of the DFT
on the ring of integers modulo $N$ is well-known by now \cite{NC}, and has been studied
for other groups as well
\cite{Beals97}.

In terms of the sampling algorithm our result generalizes, in the quantum setting, the result of Gentry et al. \cite{GPV08} to arbitrary distributions
with "nice" FT's.
In that result the authors showed how to sample from the discrete Gaussian distribution
with a variance comparable to the length of the lattice basis $\|B\|$, and here
we provide a quantum routine that can perform this task for essentially any
distribution that can be "sampled quantumly".
We note that one can also distill a quantum sampling routine from the work of Regev \cite{Reg09},
but the SysNF structure makes our scheme advantageous compared to that scheme:
we can sample quantumly from functions which are not known to be accessible
via the work of \cite{Reg09}.
We discuss this further in Section \ref{sec:sample}.

Finally, the question of sampling from general distributions on lattices has been also
investigated by Lyubashevsky and Wichs \cite{Lyu15} 
in the context of cryptographic efficiency.
There, the authors
show how to sample classically from arbitrary distributions on lattices
defined by a system of modular equations, but they also
require the knowledge of a secret trapdoor in addition to the lattice basis, in order to do that.

\subsection{Open Questions}

We believe there are several important open questions that arise from our new definition,
and its quantum implementation, that pertain to the problem of solving hard lattice problems.
One such question is trying to characterize the eigenvectors of the lattice DFT unitary:

\begin{question}
Let $L\subseteq \R^n$ be some SysNF lattice, and ${\cal F}_{L,N}$ denote its corresponding
DFT.  Find the eigenvectors of ${\cal F}_{L,N}$.
\end{question}

The interest in the above question stems from the fact that using quantum phase estimation
w.r.t. ${\cal F}_{L,N}$ and, say a randomly chosen quantum state,
it may be possible to find such eigenvectors efficiently.
On the other hand, it is known that the eigenvectors of the standard $n$-dim. DFT
are Gaussian, up to multiplying by a Hermite polynomial.
Hence it is possible that the eigenvectors of ${\cal F}_{L,N}$ are discrete Gaussian superpositions
on $L_N$.
Could it be that one of these eigenvectors is a Gaussian that is computationally "interesting"?
say with variance $s = \poly(n)$?

\section{Preliminaries}

\subsection{Notation}

The $n$-dimensional Euclidean space is denoted by $\R^n$.
The Euclidean norm of a vector $x\in \R^n$ is $\|x\| = \sqrt{\sum_{i=1}^n |x_i|^2}$.
A Euclidean lattice $L$ is written as $L = L(B)$ where $B$ is some basis of $L$.
$N$ is used to denote $\det(L)$, and $\Z_N = \mathbb{Z}/(N\mathbb{Z})$ the ring of integers modulo $N$.
Often,  we will refer to $\Z_N$ as the set of numbers $[0,\hdots,N-1]$. 
$x\modn$ is the unique value $x'$ such that $x' = x + k \cdot N$ for integer $k$, and $x'\in \Z_N$.
We define $\Delta$ as the statistical distance between distributions $(p,\Omega), (q,\Omega)$,  i.e. 
\[\Delta(p,q) = \int_{\Omega} |p(x) - q(x)| dx. \]
Given a set $S$, $U(S)$ is the uniform distribution on $S$.
For any $v\in \R^n$ define: 
$
|v| = \max_i |v_i|.
$ 
For real number $s>0$ and vector $c\in R^n$, ${\cal B}_s(c)$ is the closed Euclidean ball of radius $s$ around $c$.
%For integer $n\geq 1$, the notation $[n]$ stands for the set of indices $\{1,\hdots, n\}$.
Given a set $S\subseteq \R^n$, and a vector $v\in \R^n$, we denote $\dist(v,S): = \min_{x\in S} \|v - x\|$.
%For integer $M$, $[M] = \{1,\hdots,M\}$.
For functions $f,g$ we write $f(x) \propto g(x)$ if there exists a constant $c\neq 0$ independent of $x$ such that $f(x) = c \cdot g(x)$.

\subsection{Density of Co-Prime Numbers}

We use the following fact on the density co-primality of numbers due to Iwaniec:
\begin{fact}\label{fact:prime}\textbf{Log-density of co-prime numbers }\cite{I78}

\noindent
There exists a constant $c>0$ such that for any number $n$ with $r$ distinct prime factors
any consecutive sequence of integers of size at least
$$
c \cdot (r \log(r))^2
$$
contains an integer co-prime with $n$.
\end{fact}

\subsection{Background on Lattices}
We start by stating some standard facts about lattices.

\begin{definition}

\textbf{Euclidean Lattice}

\noindent
A Euclidean lattice $L\subseteq \R^n$ is the set of all integer linear combinations of a set of linearly independent vectors
$b_1,\hdots, b_m$:
$$
L = \left\{ \sum_{i=1}^m z_i b_i, \ \ z_i\in \Z, \right\} \subseteq \R^n
$$
This set $\{b_i\}_{i=1}^n$ is called the {\em basis} of the lattice.
We denote by $L = L(B)$, where $B$ is the matrix whose columns are $b_1,\hdots, b_m$.
In this paper, we will always assume that $L$ is full-dimensional, i.e. $m=n$.
\end{definition}
\noindent
For lattice $L = L(B)$, ${\cal P}(B)$ is the basic parallelotope of $L$ according to ${\cal B}$:
$$
{\cal P}(B) := \left\{ v= \sum_{i\in [n]} x_i b_i, \ \ x_i\in [0,1) \right\}.
$$
While the basic parallelotope of $L$ depends on the given representation of $L$ via the basis,
the Voronoi cell is a basis-independent object:
$$
{\rm Vor}(L) := \left\{ x\in \R^n, \ \ \forall y\in L, y\neq 0, \ \ \|x - y\| \geq \|x\| \right\}.
$$

\begin{definition}

\textbf{The Dual Lattice}

\noindent
The dual of a lattice $L = L(B)$ is the lattice generated by
the columns of $B^{-T}$. 
\end{definition}
%It is easy to see that the matrix $B^\perp$ is in lower triangular form. In
%particular, the dual of our lattice in Hermite normal form is the column
%space of 
%a matrix $B^\perp$ that looks
%like 
%\be
%B^\perp = \left[\begin{array}{ccccc} 
%b_{1,1}^{-1} & & & & \\
%* & b_{2,2}^{-1} & & & \\
%* & * & b_{3,3}^{-1} & & \\
%\vdots & & & \ddots & \\
%* & * &*  &\ldots  & b_{n,n}^{-1} 
%\end{array}
%\right]
%\ee
%where $*$ represents a possibly non-zero entry and the blanks represent
%zeroes. Every entry of $B^\perp$ is a fraction with denominator $D$. 
%If we define $u_i$ and $v_j$ to be the $i$th and $j$th columns of $B$ and
%$B^\perp$, we have $\langle u_i,v_j \ra = \delta_{i,j}$, because 
%$(B^\perp)^T B = (B^T B)^{-1} B^T B = I$. 

\begin{definition}

\textbf{Successive minima of a lattice}

\noindent
Given a lattice $L$ of rank $n$, its successive minima $\lambda_i(L)$ for all $i\in [n]$ are defined as follows:
$$
\lambda_i(L) = 
\inf 
\left\{ 
r | \dim( {\rm span} (L \cap \bar{B}_r(0)) ) \geq i 
\right\}.
$$
\end{definition}

\begin{definition}

\textbf{Unimodular matrix}

\noindent
The group of unimodular matrices $GL_n(\Z)$ is the set of $n\times n$
integer matrices with determinant $1$.
Unimodular matrices preserve a lattice: $L(B) = L(B')$ if and only if $B = B' \cdot A$, for some unimodular matrix $A$.
\end{definition}

\begin{definition}

\textbf{The determinant of a lattice}

\noindent
For a lattice $L = L(B)$ we define $\det(L) = \det(B)$, and denote by $N$.
\end{definition}
The determinant of a lattice is well-defined, since if $L(B') = L(B)$, then by the above $B = B' \cdot A$
for some unimodular matrix $A$, in which case $\det(B) = \det(B') \det(A) = \det(B')$.
The lattice $L$ is periodic modulo $N$. In other words, if we add $N$ to 
any coordinate of a lattice point, we reach another
lattice point. Thus, a cube of side length $N$ gives a subset of the lattice
which generates the whole lattice when acted on by translations by $N$ in
any direction. We let $L_N$ denote the lattice restricted to a cube of side
length $N$.  

In particular, if $L = L(B)$ is an integer lattice, with $\det(L) = N$ then $L_N$ is a finite additive sub-group, or lattice, of $\Z_N^n$:
\begin{proposition}\label{prop:1}
Let $\Z_N^n$ denote the additive group of $n$-dimensional vectors of integers, where in each coordinate summation is carried out modulo $N$.
Then $L_N$ is an additive sub-group of $\Z_N^n$, that contains the $0$ point.
In particular $L_N$ is a lattice of $\Z_N^n$, with $|L_N| = N^{n-1}$.
\end{proposition}
\begin{proof}
The determinant of $L$ is $N$ by definition of the systematic normal form.
Hence the size of $L_N$, which is a finite sub-group of $\Z_N^n$ is given by $|\Z_N^n| / N = N^n / N = N^{n-1}$.
\end{proof}

\noindent
A canonical representation of integer lattices is called the Hermite normal form (HNF):
\begin{definition}

\textbf{Hermite Normal Form}

\noindent
An integer matrix $A\in \Z^{n\times n}$ is said to be in Hermite normal form (HNF) if $A$ is 
upper-triangular, and $a_{i,i}> a_{i,j}\geq 0$ for all $j>i$, and all $i\in [n]$.
\end{definition}

\noindent
It is well-known that every integer matrix can be efficiently transformed into HNF:
\begin{fact}\textbf{Unique, efficiently-computable, Hermite normal form }\cite{KB79}

\noindent
For every full-rank integer matrix $A\in Z^{n\times n}$, there exists a unique unimodular
matrix $U\in GL_n(\Z)$, such that $H = U \cdot A$, and $H$ is HNF.
$U$ can be computed efficiently.
\end{fact}

%\begin{definition}
%
%\textbf{Lattice covering radius}
%
%\noindent
%Let $L\subseteq \R^n$ be some lattice.
%The covering radius of $L$, $\rho(L)$ is the minimal number such that any $x\in \R^n$
%is at distance at most $\rho(L)$ from $L$.
%\end{definition}

%The work of Banaszczyk \cite{Ban93} established important connections 
%between parameters of a lattice and its dual:
%\begin{fact}\label{fact:ban1}
%For any lattice $L\subseteq \R^n$ we have
%$$
%\forall i\in [n], \ \ 1 
%\leq \lambda_i \lambda_{n-i+1}^* \leq n
%$$
%\end{fact}
%
%\begin{fact}\label{fact:cover1}
%For any $L\subseteq \R^n$ we have $\rho(L) \cdot \lambda_1^* \leq n$.
%\end{fact}

The following proposition, due to Babai \cite{Bab86}, builds on the famous LLL algorithm and shows that one can solve the closest vector
problem up to an error that is at most exponential in the dimension:
\begin{proposition}\label{prop:np1}

\textbf{The Nearest-Plane Algorithm }\cite{Bab86}

\noindent
There exists an efficient algorithm ${\cal M}$ such that 
for any $u\in \R^n$ and lattice $L = L(B)\subseteq \R^n$
the vector $v = {\cal M}(u,B)\in L$ satisfies:
$$
\| u - v \| \leq 2^{n/2} \cdot \dist(u,L).
$$
\end{proposition}

\section{The Systematic Normal Form (SysNF)}\label{sec:snf}

In this section we explore the definition Systematic Normal Form introduced in \cite{ES16} and 
discuss some of its basic properties
\noindent
The following facts will be useful later on.
First, by simple matrix in version one obtains:
\begin{proposition}\label{prop:ndual}
If $B$ is SysNF form, then $NB^{-T}$, i.e. the matrix spanning the scaled dual of $L(B)$ assumes the following form:
\be\label{eq:ndual}
N \cdot B^{-T} = \left[\begin{array}{ccccc} 
1 & &&& \\
-b_{2} & N& & &\\
-b_{3} & & N & &  \\
\vdots & & & \ddots & \\
-b_{n} & & & & N 
\end{array}
\right]
\ee
\end{proposition}

\noindent
In the paper,
we will use the following notation:
$$
L_N := L \cap \Z_N^n,
\quad
(NL^*)_N := (N \cdot L^*) \cap \Z_N^n
$$

\begin{proposition}
Let $L$ be a SysNF lattice with $N = \det(L)$.
Then $|(NL^*)_N| = N$, and $|{\cal P}(L) \cap \Z_N^n| = N$.
Both $L, NL^*$ are periodic in $N$ -- i.e. $N e_i\in L, N e_i\in (NL^*)_N$ for every $i\in [n]$.
\end{proposition}
\begin{proof}
By Proposition \ref{prop:ndual} $(NL^*)_N$ is a cyclic group of order $N$, hence its size is $N$.
By definition of SysNF we have that $\det(L) = N$, hence $L$ is periodic in $N e_i$ for all $i\in [n]$.
In terms of the dual $NL^*$ each of the last $n-1$ columns is a multiple $e_i N$,
and the vector $N e_1$ is achieved by adding suitable multiples of $N e_i$ for $i>1$
to the vector $(N B^{-T}) \cdot (N e_1)$.
\end{proof}

We now state the following important property:
\begin{claim}\label{cl:dual}

\textbf{Efficient bijection between quotient and dual}

\noindent
There exists an efficiently-computable bijection $\Phi_3: \Z_N^n / L_N \mapsto (NL^*)_N$,
such that for every $x\in \Z_N^n$, $x + \Phi_3(x) \in L_N$.
\end{claim}
\begin{proof}
Let $x\in \Z_N^n$.
We want to find (the unique) $y = \Phi_3(x)$ for which $x + y\in L_N$.
Each point in $y\in (NL^*)_N$ is characterized uniquely by an element $a\in \Z_N$ as follows:
\be\label{eq:nl}
y = (a, -b_2 a  \modn,\hdots, -b_n a\modn).
\ee
Thus, to find $y$ we solve the following vector equality over $a,z_2,\hdots, z_n\in \Z_N$:
\be
(x_1,\hdots, x_n)^T + (a, -b_2 a \modn,\hdots, -b_n a \modn)^T
=
\left(\sum_{i=2}^n b_i z_i \modn, z_2,\hdots, z_n\right)^T
\ee
Consider the first coordinate. We have:
\be
x_1 + a = \sum_{i=2}^n b_i z_i \modn.
\ee
Substituting in the above $z_i = x_i - a b_i \modn$ for all $i\geq 2$ implies:
\be
x_1 - \sum_{i=2}^n x_i b_i =  - a \cdot \left( \sum_{i=2}^n b_i^2 + 1 \right) \modn.
\ee
Since $\sum_{i=2}^n b_i^2 + 1$ is co-prime to $N$ then it has an inverse modulo $N$.
Thus, the parameter $a$ can be
computed uniquely from the equation above, which implies that $y$ can be determined uniquely and efficiently.
\end{proof}

\subsection{Reduction to SysNF}

In this section we provide an efficient reduction from an arbitrary lattice to a lattice in SysNF form,
that preserves all important properties of the lattice.
Specifically,
it allows the reduction of any computational problem on an arbitrary lattice $L$ to another problem on an SysNF lattice $L_{SysNF}$
such that any solution to the reduced problem allows one to find {\it efficiently} a solution to the original problem on $L$.

\begin{lemma}\label{lem:SysNF}

\textbf{Efficient reduction to SysNF}

\noindent
There exists an efficient algorithm that for any $L = L(B)$ and $\eps>0$
computes a tuple $\langle B',\sigma,T \rangle$, where
$B'$ is SysNF, $T = \poly(\det(B)/\eps)$ is a positive integer
and $\sigma$ is a linear map $\sigma: L \to L(B')$ such that
 for any $v\in L$ we have $\|\sigma(v)/T - v \| \leq \|v\| \eps$.
\end{lemma}
The lemma above implies that one can reduce standard lattice problems, given for an arbitrary lattice, to the same problem on a lattice in SysNF,
and then translate the output solution efficiently to a solution for the original lattice.

\noindent
Before presenting the proof,
let us bound the coefficients of any short vector in a lattice. 
\begin{proposition}\label{prop:lll1}
Let $B$ be some matrix, and $v\in L$ be some lattice vector.
Then $v$ can be represented in the basis $B$ using a vector coefficients of absolute value at most $ \|v\| \det(B)$.
\end{proposition}
\begin{proof}
Follows immediately by Cramer's rule: $v = B \cdot z$, then the magnitude of each $z_i$ is at most $\det(B_i) / \det(B)$,
where $B_i$ is the matrix derived from $B$ by replacing the $i$-th row with $v$.
\end{proof}

\noindent
The following is an easy corollary of the above:

\begin{proposition}\label{prop:close}
\label{two-close-lattices}
Let $B_1 = \{v_i\}_{i=1}^n$ be some basis
and another basis $B_2 = \{w_i\}_{i=1}^n$ for lattice $L_2 = L(B_2)$.
Suppose that $\|v_i - w_i\| \leq \alpha$. 
Let $v = \sum_{i=1}^n c_i v_i$ be a point in $L_1$ and $w = \sum_{i=1}^n
c_i w_i$ be the corresponding point in $L_2$. Then $\left\|v - w \right\| \leq n\|v\| \alpha  \det(B_1) $. 
\end{proposition}

\begin{proof}
By the triangle inequality we have:
\begin{eqnarray}
\left\|v-w \right\| &= &  \left\|\sum_{i=1}^n c_i v_i  -  \sum_{i=1}^n c_i w_i  \right\| \\
&\leq & \sum_{i=1}^n \| v_i -w_i \| |c_i| \\
&\leq & n\alpha \|v\| \cdot \det(B_1)
\end{eqnarray}
where the last inequality follows from Proposition \ref{prop:lll1}. 
\end{proof}

\subsection{Proof of Lemma \ref{lem:SysNF}}

\begin{proof}
We first use $T$ as a parameter and determine it later on in the proof.
We start from an upper-triangular matrix $B_1$ in Hermite normal form:
\be
B_1 = \left[\begin{array}{ccccc} 
b_{1,1} & b_{1,2} & b_{1,3} & \ldots & b_{1,n} \\
& b_{2,2} & b_{2,3} & \ldots & b_{2,n} \\
& & b_{3,3} & \ldots & b_{3,n} \\
& & & \ddots & \vdots \\
& & & & b_{n,n} 
\end{array}
\right]
\ee
add $1/{T}$ along the sub-diagonal,
and truncate each non-zero entry to its nearest integer multiple of $1/T$:
\be
\label{add-subdiagonal}
B_2 = \left[\begin{array}{ccccc} 
b_{1,1}' & b_{1,2}' & b_{1,3}' & \ldots & b_{1,n}' \\
\frac{1}{T}& b_{2,2}' & b_{2,3}' & \ldots & b_{2,n}' \\
& \frac{1}{T}& b_{3,3}' & \ldots & b_{3,n}' \\
& & \ddots & \ddots & \vdots \\
& & & \frac{1}{T}& b_{n,n}'
\end{array}
\right],
\ee
where $b_{i,j}' = \round{b_{i,j} T} /T$.
We note that
\be\label{eq:entry}
\forall i,j \ \ |B_2(i,j) - B_1(i,j)| \leq 1/T.
\ee

We now use column operations to make rows $2$, $3$, $\ldots$,
$n$ of the lattice zero except for the sub-diagonal. This involves subtracting integer multiples
of the $i$th column from all later columns. We obtain a lattice of the form.
\be\label{eq:b3}
\label{after-Gaussian}
B_3 = \left[\begin{array}{ccccc} 
b_{1,1}' & b_{1,2}'' & b_{1,3}'' & \ldots & b_{1,n}'' \\
\frac{1}{T}& 0 & 0 & \ldots & 0 \\
& \frac{1}{T}& 0 & \ldots & 0 \\
& & \ddots & \ddots & \vdots \\
& & & \frac{1}{T}& 0 
\end{array}
\right]
\ee

We now compute the matrix that transforms the basis given in equation
(\ref{add-subdiagonal}) to the basis given in equation (\ref{after-Gaussian}).
That is, we want the matrix $M$ such that $B_3 =  B_2 M$. The diagonal and
super-diagonal of the matrix can be easily calculated: 
\begin{eqnarray}
M &=& \left( 
\begin{array}{ccccccc}
1 & T b_{2,2}' & T b_{2,3}' &  \ldots & T b_{2,n-1}'&   T b_{2,n}' \\
  & 1 & T b_{3,3}' & \ldots & T b_{3,n-1}' & T b_{3,n}' \\
  &   & 1 & \ldots & T b_{4,n-1}' & T b_{4,n}' \\
  &      &   & \ddots &   \vdots & \vdots \\
  &   & & & 1 & T b_{n,n}'  \\ 
  &   &  & &   & 1 
\end{array}
\right)^{-1} \\
&=&
\left( 
\begin{array}{ccccccc}
1 & - T b_{2,2}' &   \ldots &   \\
  & 1 & - T b_{3,3}' & \ldots &  \\
  &   & 1 & - T b_{4,4}' & \ldots\\
  &      &   & \ddots &   \ddots  \\
  &   & & & 1 & - T b_{n,n}'  \\ 
  &   &  & &   & 1 
\end{array}
\right)
\end{eqnarray}
By the above, $M$ is a unimodular matrix, with $\det(M)=1$, and hence $\det(B_2) = \det(B_3)$.
Note that both $M$ and $M^{-1}$ are upper triangular matrices with $1$s along the diagonal.
The determinant of $TB_3$ is $T b_{1,n}''$. 

Observe that if we move the $n$th column of $T B_3$ to the first column we get a lattice which a SysNF lattice, 
except possibly from the entry $b_{1,n}''$ which may not satisfy the condition \ref{eq:cond1}.
Using Fact \ref{fact:prime} there exists an integer $0<\delta \leq c \cdot \log^3( T b_{1,n}'')$ such that:
\be\label{eq:delta}
\sum_{j=1}^{n-1}  (T b_{1,j}'')^2+1 \neq 0 (\mod (T b_{1,n}'' + \delta)).
\ee
By enumerating over all numbers from $T$ to $T + \delta$ and
invoking the Euclidean algorithm for each, we can find such a number $\delta$ efficiently.

So now we modify $T B_3$ by adding $\delta$.
This corresponds to adding $\delta/T$ to the entry $B_3(1,n)$. 
What effect does this change have on the basis of the lattice in $B_2$?
Let $\Delta$ be the matrix with
$\Delta(1,n) = \delta/T$ and all other entries $0$. Then our matrix in the SysNF basis
is $B_3 + \Delta$. To see what the effect on $B_2$ is, we merely need to multiply by
$M^{-1}$. That is,
\be
B_2 + \Delta M^{-1} = ( B_3 + \Delta )M^{-1} .
\ee
Using the form we derived above for $M^{-1}$, we see that because there are 1s along the
diagonal of $M$ then 
$\Delta M^{-1} = \Delta$. Thus, we can make $B_2$ have a determinant satisfying the condition above by simply
adding $\delta/T$ to $B_2(1,n) = b_{1,n}'$. This
changes the length of the $n$th basis vector by at most
$\delta/T$.

Let $B_4$ denote then the output SysNF matrix.
\be
B_4 = \left[\begin{array}{ccccc} 
T b_{1,n}'' + \delta & T b_{1,1}' & T b_{1,2}'' & \ldots & T b_{1,n-1}'' \\
0 & 1 & 0 & \ldots & 0 \\
   & 0 & 1 & \ldots & 0 \\
& & \ddots & \ddots & \vdots \\
& & & 0& 1 
\end{array}
\right]
\ee
By Equations \ref{eq:entry} and equation \ref{eq:delta} :
\be
\forall i,j\in [n] |(M^{-1} B_4(i,j))/T - B_1(i,j)| = 2c \log(T b_{1,n}'')/T.
\ee
That is, the basis $M^{-1} B_4/T$ of $L(B_4)/T$ is entry-wise close to $B_1$.
We invoke Proposition \ref{prop:close} w.r.t. these two bases.
Consider some $v\in L(B_4/T)$.
Applying Proposition \ref{prop:close} implies that 
the 
corresponding vector $\hat{v} = B_1 \cdot (TB_4^{-1}M) \cdot v\in L(B_1)$
has
\be\label{eq:bound2}
\left\|\hat{v}- v\right\| \leq  \frac{n \det(B_1) \|v\| 2c \log^3(T b_{1,n}'')}{T}.
\ee
By Equation \ref{eq:bound2} we conclude that there exists
\be
T = \poly(\det(B_1)/\eps)
\ee
such that
\be
\left\|\hat{v}- v\right\| \leq  \|v\| \eps.
\ee
Setting $\sigma: = M^{-1} B_4$ fixes a bijection $L \to L(B_4')$ with the aforementioned property.
By our choice of $\delta$ we have that $T b_{1,n}'' + \delta$ and $\sum_{j=1}^{n-1} {T b_{1,j}''}^2+1$ are co-prime.
Hence, $B_4$ satisfies condition \ref{eq:cond1} and so it is a valid SysNF matrix.

\end{proof}

\section{A Discrete Fourier Transform on SysNF lattices}\label{sec:fourier}

\subsection{Defining the Discrete Fourier Transform}
The Fourier Transform on Euclidean lattices is normally associated with the
Fourier series of lattice-periodic functions:
Let $L\subseteq \R^n$ denote some full-rank $n$-dimensional lattice, 
with a basic parallelotope ${\cal P}(L)$.
Consider the set of bounded complex-valued continuous functions $f: \R^n \to \C$
that are periodic in $L$, i.e.
$$
\forall x\in \R^n, z\in L, \ \ f(x) = f(x + z).
$$
Then the Fourier series of $f$, $\hat{f}: L^* \mapsto \C$, supported on the dual lattice $L^*$ is defined as follows:
$$
\forall z\in L^*, \ \ 
\hat{f}(z) := \frac{1}{\det(L)} \cdot \int_{{\cal P}(L)} f(x) e^{-2\pi i \langle x ,z \rangle} dx.
$$
One can invert the Fourier series by:
$$
\forall x\in \R^n , \ \ f(x) = \sum_{z\in L^*} \hat{f}(z) \cdot e^{2\pi i \langle x,z \rangle}
$$

%In the case of SysNF lattices $L_N$ is an Abelian sub-group
%of $\Z_N^n$ and hence admits a discrete Fourier Transform: namely
%the character table of $1$-dimensional irreducible representations of $L_N$ into characters of the unit group.
%It is however, unclear, a-priori whether such a representation is indeed the natural
%$n$-dimensional generalization of the $1$-dimensional discrete Fourier Transform:

Recall that the $1$-dimensional discrete Fourier transform is defined as follows:
let $N$ be some integer, and let $f: \Z_N \mapsto \C$ denote the set of real functions on $\Z_N$, and
$\Z_N = \{0,\hdots,N-1\}$.
Then
$$
\forall z\in \Z_N, \ \ 
{\cal F}(z) = \sum_{x\in \Z_N}  f(x) e^{-2\pi i  x \cdot z / N}.
$$
with a similar inversion:
$$
\forall x\in \Z_N, \ \ 
f(x) = \frac{1}{N} \sum_{z\in \Z_N}  {\cal F}(z) e^{-2\pi i  x \cdot z  / N}.
$$
Similarly one can define the DFT on $n$-dimensional functions $f: \Z_N^n \to \C$ as follows:
$$
\forall z\in \Z_N^n, \ \ 
{\cal F}(z) = \sum_{x\in \Z_N^n}  f(x) e^{-2\pi i  \langle x,z\rangle / N}.
$$
We would like a generalization of the DFT to arbitrary lattices.
We define:
\begin{mdframed}
\begin{definition}

\textbf{Discrete Fourier Transform on $L_N$}

\noindent
Let $L = L(B)$ be a lattice spanned by SysNF matrix $B$.
The Discrete Fourier Transform on $L_N$, i.e. ${\cal F}_{L,N}$ is defined as follows:
\be
\forall f\in \R^{L_N} \ \ 
\forall x\in L_N \ 
{\cal F}_{L,N}(x) = \sum_{z\in L_N} f(z) e^{-2\pi i \langle x,z \rangle / N}.
\ee
\end{definition}
\end{mdframed}

\noindent
\\
The number-theoretic property, namely that $\sum_{i>1} B_{1,i}^2 \neq (-1) \pmod{N}$ allows
us to show, using a quantum argument, that the character table of ${\cal F}_{L,N}$ is a unitary matrix. 

\paragraph{Discussion:}
This implies that for any $x\neq y$ the functions $\chi_x(z) = e^{-2\pi i \langle x,z\rangle/N}$ and $\chi_y(z) = e^{-2\pi i \langle y,z\rangle/N}$, regarded
as characters of the group $L_N$ are orthogonal.
Since for each $x$ the character function $\chi_x(z)$ is a homomorphism
$L_N \to \C^{|L_N|}$, and the number of these character functions is also precisely $|L_N| = N^{n-1}$,
then the set of character functions form a complete set of inequivalent irreducible representations of $L_N$.
Therefore,  ${\cal F}_{L,N}$ is a Discrete Fourier Transform of $L_N$, as a Fourier Transform 
over the finite group $L_N$ using a set of irreducible representations on the cyclic group modulo $N$.

We note that for any lattice $L$, with $\det(L)=N$ the set $L_N$ is in particular a {\it finite} Abelian group, and hence is isomorphic
by the Fundamental Theorem of Finite Abelian group to a direct product of $\Z_{p_k}^k$
where $p_k$ is prime.
This then gives rise to a natural Fourier Transform as a direct product of the Fourier Transform on $\Z_{p_k}^k$ for each factor $k$.
However, for a typical lattice $L$ where $N = \det(L)$ is not a prime number, 
it is unclear how one would efficiently find the isomorphism between $L_N$ and its factors, and even if so - whether it would
amount to a DFT - i.e. have the characters correspond to the integer inner-product modulo $N$.

%A possibly elucidating example of this is the lattice $L\subseteq \R^n$ spanned by $p_i e_i$, where the $p_i$'s are unique prime numbers.
%Then $N = \det(L) = \prod_{i=1}^n p_i$.
%By the above, we can define a DFT on $L_N$ as the DFT on the product $\Z_{p_1} \times \hdots \times \Z_{p_n}$.
%Given any basis $B$ of $L$, one can find the order of these $p_i$'s by factoring $N$ (using a quantum computer, say), and hence determine its group structure.
%However, given an arbitrary basis it may be very hard to implement the Fourier Transform over $L_N$.
%For example,  suppose that $B$ contains the basis element $b_1 = \sum_i p_i e_i$.
%Attempting to use the columns of $B$ as generators of the direct product group directly would be impossible in this case,
%since the order of $b_1$ is equal to $lcm(p_1,\hdots,p_n) = N$, and in fact, one can easily
%define a basis in which each element is of order $N$ individually.

\subsection{An Efficient Quantum Algorithm}

So first, just like the standard QFT is a quantum implementation of the $1$-dimensional DFT,
we define a quantum DFT map on lattices:
\begin{mdframed}
\begin{definition}

\textbf{Quantum Fourier Transform on SysNF lattices}

\noindent
Let $L\subseteq\R^n$ be a SysNF lattice, $N = \det(L)$.
The Quantum Fourier Transform on $L_N$ is defined for basis states as follows:
\be
\forall x\in L_N, {\cal F}_{L,N}(\ket{x}) = \frac{1}{\sqrt{N^{n-1}}}\sum_{z\in L_N} e^{-2\pi i \langle x,z \rangle/N} \ket{z}.
\ee
\end{definition}
\end{mdframed}

\noindent
\\
Next, we show that the above map can be implemented efficiently using a quantum circuit.
This, in particular, establishes that ${\cal F}_{L,N}$ is orthogonal, and hence qualifies as a DFT of $L_N$:
\begin{theorem}\label{thm:qft}

\noindent
Given is a lattice $L = L(B)$, where $B$ is an $n\times n$ SysNF matrix.
There exists a quantum circuit ${\cal Q}$ of size $\poly(n)$, that implements ${\cal F}_{L,N}$.
In particular, ${\cal F}_{L,N}$ is a unitary matrix, and hence it is the DFT of $L_N$.
\end{theorem}

\begin{proof}

We are given a lattice $L$ represented by an SysNF matrix $B$, where $N = \det(B)$.
All arithmetic computations are carried out w.r.t. the ring $\Z_N$.
Write: 
\be
\forall x\in L_N \ 
\ket{x} = \ket{x_1} \otimes \hdots \otimes \ket{x_n}
\ee
The matrix $B$ is parameterized by $b_{1,1} = N, b_{1,2}, \hdots, b_{1,n}$.
For simplicity of notation, put $b_{1,j} = b_j$.
Compute:
\be
\ket{x} 
\to
\ket{x_1} \otimes \ket{x_2 + b_{2} x_1} \otimes \hdots \otimes \ket{x_n + b_{n} x_1}.
\ee
We now claim that one can un-compute $\ket{x_1}$.
Let 
\be
\forall j, 2 \leq j \leq n \ \ \ 
y_j = x_j + b_{j} x_1
\ee
i.e. our register is
\be
\ket{x_1} \otimes \ket{y_2} \hdots \ket{y_n}.
\ee
Let
\be
\phi(x) = \sum_{j=2}^n x_j b_j.
\ee
Then by definition of $B$, and the fact that $x\in L(B)$ we have that
\be
x_1 = \phi(x) + s \cdot N,  s\in \mathbf{Z}.
\ee
Therefore
\begin{align}
\sum_{j=2}^n b_j y_j
&=
\sum_{j=2}^n b_j x_j + x_1 \sum_{j=2}^n b_j^2 \\
&=
x_1 + s \cdot N + x_1 \sum_{j=2}^n b_j^2 \\
&=
x_1\cdot (\sum_{j=2}^n b_j^2 + 1) + s \cdot N.
\end{align}
Suppose that $x_1 \neq 0 \modn$.
Then since by assumption 
\be
\sum_{j=2}^n b_j^2 +1 \neq 0 \modn
\ee
then
\be
x_1 \modn = \left(\sum_{j=2}^n b_j^2 + 1\right)^{-1} \left(\left(\sum_{j=2}^n b_j y_j\right) \modn\right),
\ee
where the existence of the inverse is implied by condition \ref{eq:cond1}.

Therefore, using only $y_2,\hdots,y_n$ we can compute $x_1$ up to $s \cdot N$, for some $s\in \Z$.
Since by definition $x\in L_N$, then $x_1\in \Z_N$, so we can determine $x_1$ exactly.
Hence, we map unitarily:
\be
\ket{x_1} \otimes \ket{y_2} \hdots \ket{y_n}
\to
\ket{y_2} \hdots \ket{y_n}
\ee

Now, we apply $n$ tensor-product copies of the standard $1$-dimensional QFT on $N$ points.
We get:
\begin{align}
\ket{y_2} \otimes \hdots \otimes \ket{y_n}
&\to
\frac{1}{\sqrt{N^{n-1}}}
\left(
\sum_{z_2\in \Z_N} e^{-2\pi i \langle y_2,z_2 \rangle/N} \ket{z_2}
\right) 
\hdots
\left(
\sum_{z_n\in \Z_N} e^{-2\pi i \langle y_n,z_n \rangle/N} \ket{z_n}
\right) \\
&=
\frac{1}{\sqrt{N^{n-1}}}\sum_{z \equiv (z_2,\hdots,z_n) \in \Z_N^{n-1}}
e^{-2\pi i \langle (y_2,\hdots,y_n),z \rangle /N} \ket{z} \\
&=
\frac{1}{\sqrt{N^{n-1}}}\sum_{z_2,\hdots,z_n \in \Z_N^{n-1}}
e^{-2\pi i  x^T B z' / N} \ket{z} 
\end{align}
where $z' = z'(z)\in \Z_N^n$ is some vector for which $z_i' = z_i$ for all $i>1$.
Hence the above is equal to
\be
\frac{1}{\sqrt{N^{n-1}}}\sum_{z_2,\hdots,z_n \in \Z_N^{n-1}}
e^{-2\pi i  \langle x, Bz' \rangle / N} \ket{z}
\ee
Let us now apply the matrix $B$ unitarily:
\be
\forall z\in \Z_N^{n-1}
\ \ 
\ket{z}
\to
\ket{ B z \modn} =
\ket{ (\sum_{i>1} b_i z_i, z_2,\hdots, z_n)}.
\ee
The RHS is a function of $z_i$ for $i>1$, but independent of $z_1$.
Hence
\be
\ket{B z \modn} = \ket{B z' \modn}
\ee
Therefore, we get the state
\begin{align}
\frac{1}{\sqrt{N^{n-1}}}\sum_{z\in \Z_N^{n-1}} e^{-2\pi i \langle x,Bz'\rangle / N} \ket{Bz'}
&=
\frac{1}{\sqrt{N^{n-1}}}\sum_{w\in L_N} e^{-2\pi i \langle x,w\rangle / N} \ket{w}.
\end{align}
\end{proof}
We note that the above statement only claims that for any $x\in L_N$ the state $\ket{x}$ can be mapped
unitarily to the character of $x$ on $L_N$.
As a quantum circuit on $n$ coordinates one would additionally need to specify the action
when $x\notin L_N$. 
One such possibility is simply to apply the identity map for all such $x$.

\subsection{Properties of the QFT}

\noindent
It is straightforward to check that the phase-shift, and linear-shift unitary matrices
are a conjugate pair w.r.t. ${\cal F}_{L,N}$ when the shift is by a lattice vector:
\begin{proposition}\label{prop:phase1}

\textbf{Vector-shift -- phase-shift equivalence}

\noindent
$$
\forall x\in L_N, v\in L_N, \ \ 
{\cal F}_{L,N} \circ U_v \ket{x} = W_v \circ {\cal F}_{L,N} \ket{x}
$$
\end{proposition}

\noindent
The following fact also follows immediately from the definition of the DFT:
\begin{proposition}\label{prop:ft1}
For a function $f: \Z_N^n \to \R$, let $\hat{f}$ denote it's $n$-dimensional $N$-point DFT:
$$
\forall x\in \Z_N^n
\ \ 
\hat{f}(x):= \sum_{z\in \Z_N^n} f(z) e^{-2\pi i \langle z,x\rangle/N}.
$$
for any function $f: \Z_N^n \to \R$
the quantum state
$$
\ket{\tilde f} \propto {\cal F}_{L,N} \cdot \sum_{x\in L_N} f(x) \ket{x},
$$
may be written as:
$$
\ket{\tilde f} \propto \sum_{x\in L_N} \hat{f}(x) \ket{x},
$$
\end{proposition}

\noindent
In particular, when $f$ satisfies a certain "smoothness" condition -- this restriction
is roughly periodic around $NL^*$ as follows:
\begin{fact}\label{fact:sinc1}

\textbf{QFT of functions with smooth FT}

\noindent
Let $L = L(B)$ be a SysNF lattice, $N = \det(B)$.
Let $f: \Z_N^n \mapsto \R$ be some real-valued function.
Let $\hat{f}$ denote the $n$-dimensional $N$-point DFT of $f$.
Suppose that $\hat{f}$ is 
square-integrable, i.e. $\sum_{x\in L} \hat{f}(x)^2 < \infty$ and it is
$\eps$-FT-smooth, i.e. for every $v\in Vor(L)$ we have
$$
\hat{f}_{L+v}^2(L) \equiv
\sum_{x\in L} (\hat{f}(x-v))^2 
\geq
(1 - \eps)
\cdot
\sum_{x\in L} \hat{f}(x)^2 \equiv \hat{f}_L^2(L). 
$$ 
For any function $f$ let $\ket{f}_L \propto \sum_{x\in L} f(x) \ket{x}$,
and $\ket{f}_{L_N}$ denote the restriction of $\ket{f}$ to the points of $L_N$.
Then
$$
{\cal F}_{L,N} \ket{f}_{L_N}
\propto
\left(\sum_{y\in NL^*} \ket{\hat{f}_y}_L\right)_{L_N}
+ \ket{\cal E},
$$
where 
$\hat{f}_y(x) = \hat{f}(x - y)$, 
$\| \ket{\hat{f}_y}_L \| = 1$ and $\|\ket{\cal E}\| \leq \eps$.
\end{fact}

\begin{proof}

Let $g: L_N \mapsto \R$ denote the function that describes the amplitudes of ${\cal F}_{L,N} \ket{f}$:
$$
{\cal F}_{L,N} \ket{f}_{L_N} = \sum_{x\in L_N} g(x) \ket{x}.
$$
For each $x\in L_N$ we have by definition:
$$
{\cal F}_{L,N} \ket{x} =\frac{1}{ \sqrt{N^{n-1}}}\sum_{z\in L_N} e^{-2\pi i \langle x,z \rangle / N} \ket{z}.
$$
Hence $\ket{x}$ is mapped to a super-position over the $\Z_N^n$ character of $x$, {\it restricted} to the lattice $L_N$.
For a function $f: \Z_N^n \mapsto \R$ let $f_L: \Z_N^n \mapsto \R$ denote the function:
\[
  f_L(x) = \left\{\def\arraystretch{1.2}%
  \begin{array}{@{}c@{\quad}l@{}}
    f(x) & \text{if $x\in L$}\\
    0 & \text{o/w}\\
  \end{array}\right.
\]
Then by the above, the amplitudes of ${\cal F}_{L,N} \ket{f}$, namely $g(x)$, are given by computing
the full $\Z_N^n$-DFT of $f_L$, and then restricting to $L$: 
\be\label{eq:g1}
\forall x\in L_N, \ \ 
  g(x) =
    \hat{f_L}(x)  
\ee
Since $f_L$ is supported on $L$ then $\hat{f_L}$ is periodic on $NL^*$ as follows:
\be\label{eq:f1}
\forall x\in \Z_N^n, \ \ 
\hat{f_L}(x) = \sum_{y\in NL^*} \hat{f}(x-y),
\ee
Equations \ref{eq:f1}, \ref{eq:g1} imply together that:
\be
\forall x\in L_N, \ \ 
g(x) \propto \sum_{y\in NL^*} \hat{f}(x-y) \equiv \sum_{y\in NL^*}  \hat{f}_y(x),
\ee
and since $\hat{f}$ is square integrable then $\hat{f}_y$ is square-integrable for each $y$, so we can re-write the above as:
\be\label{eq:sum1}
\ket{g}_{L_N} \propto   \sum_{x\in L_N} \sum_{y\in NL^*}\hat{f}_y(x) \ket{x}
=
\left(\sum_{y\in NL^*}
\sqrt{\sum_{x\in L} \hat{f}_y(x)^2}
\ket{\hat{f}_y}_L\right)_{L_N},
\ee
where $\ket{\hat{f}_y}_{L}$ is a normalized state.
By the smoothness assumption we have that
\be\label{eq:y1}
\forall v\in {\rm Vor}(L), \ \ \sum_{x\in L} \hat{f}(x-v)^2 \geq  (1 - \eps) \sum_{x\in L} \hat{f}(x)^2.
\ee
Since ${\rm Vor}(L) \cong {\cal P}(L) \equiv \Z_N^n / L_N \cong (NL)^*_N$ then
\be\label{eq:y1}
\forall y\in NL^*, \ \ \sum_{x\in L} \hat{f}(x-y)^2 \geq  (1 - \eps) \sum_{x\in L} \hat{f}(x)^2.
\ee
Plugging back into Equation \ref{eq:sum1} implies that the state is close
to a uniform super-position over shifted copies of $\hat{f}$, namely $\hat{f}_y$:
\be
\ket{g}_{L_N} 
\propto 
\left( \sum_{y\in NL^*} \ket{\hat{f}_y}_L \right)_{L_N}+ {\cal E},
\quad
\|{\cal E}\| \leq \eps
\ee
\end{proof}

\section{Sampling Functions with "Nice" FT's}\label{sec:sample}

We now consider the problem of sampling from lattices.
One usually considers a distribution ${\cal D}$ on $\R^n$, and then asks whether
we can sample from the discrete distribution ${\cal D}$ restricted to $L$, 
i.e. where each $x\in L$ is sampled with probability proportional to ${\cal D}(x)$.
Notably, for certain distributions ${\cal D}$, sampling from the discrete distribution of ${\cal D}$ on an arbitrary
lattice $L$ is at least as hard as solving some version of the shortest vector problem on $L$.
For example, if we can sample from ${\cal D}(x) \propto e^{-\pi \|x\|^2 / s^2}$, where $s$ is comparable to $\lambda_1(L)$,
then w.h.p. we sample a lattice vector $s$ of length at most $c \lambda_1(L) \sqrt{n}$, for some constant $c>0$,
thereby solving an approximate version of the shortest-vector problem (SVP), that has no known efficient solution.

An interesting question though, given a lattice basis $B$, is whether
or not we can sample from, say, the Gaussian distribution $e^{-\pi \|x\|^2 / s^2}$, with
$s$ as small as possible given other SVP algorithms.
In \cite{GPV08} the authors have provided an affirmative answer to this question, showing a classical
algorithm that can sample from the discrete Gaussian on any lattice for all $s$ at least  $2^{n/2} \lambda_1(L) \ln(n)$,
i.e. almost matching the the bound provided by the LLL algorithm of $2^{n/2} \lambda_1(L)$.
The sampling algorithm of \cite{GPV08} relies crucially on the fact that the desired distribution
is the $n$-th fold product of the Gaussian measure.

In this work, we would like to extend this result to a more general class of distributions using quantum circuits.
Using the quantum circuit in Lemma \ref{thm:qft} one can derive a method to sample from distributions on lattices,
whose Fourier Transforms are efficiently samplable in the quantum sense.

We require several additional definitions.
The first one defines a quantum analog of what it means for a distribution to be efficiently samplable:
\begin{definition}

\textbf{Quantumly- Efficiently samplable functions}

\noindent
A function $f: \Z_N^n \to \R$ is said to be QES (quantumly efficiently samplable) if
there exists a quantum circuit of size $\poly(n)$ that generates the quantum state
$$
\frac{1}{\sqrt{\sum_{x\in \Z_N^n} |f(x)|^2}} \cdot
\sum_{x\in \Z_N^n} f(x) \ket{x}
$$
\end{definition}

\noindent
The second one relates to functions whose square-measure is bounded in some $n$-dimensional ball
in real space, at least approximately:
\begin{definition}

\textbf{Bounded functions}

\noindent
For $\eps \in (0,1)$ and $s>0$ a function $f: \Z_N^n \to \R$ is $(\eps,s)$-bounded if 
$$
\sum_{x\in \Z^n \cap {\cal B}_s(0)} |f(x)|^2 
\geq (1 - \eps) 
\sum_{x\in \Z^n} |f(x)|^2
$$
\end{definition}

\noindent
Borrowing from computational learning theory, we define a notion of "PAC" (probably approximately correct) sampling.
We consider a distribution  ${\cal D}$ on real space 
and sample lattice points proportionally to ${\cal D}(x)$ for each $x\in L_N$.
We then want to approximate ${\cal D}$: we allow both a statistical error (probably)
and a Euclidean error (approximately):
\begin{definition}

\textbf{PAC sampling}

\noindent
Let $L\subseteq \R^n$ be some lattice, ${\cal D}: \R^n \to [0,1]$ some probability measure on $\R^n$.
${\cal D}'$ is an $(\eps,\delta)$ PAC sampler for ${\cal D}$ w.r.t. $L$ if there exists a nearby probability measure ${\cal D}''$,
$\Delta( {\cal D}', {\cal D}'') \leq \eps$ satisfying:
$$
\forall x\in L \quad
{\cal D}''(x) = {\cal D}(x + \eps(x)), \ \ \|\eps(x)\| \leq \eps.
$$
\end{definition}

We present the following quantum algorithm, that invokes the nearest plane algorithm ${\cal M}$
using the basis $B$, and our specialized SysNF-QFT.
The entire algorithm is thus encapsulated by a SysNF (and inverse SysNF) map:
\begin{mdframed}
\begin{algorithm}
${\rm Sample}(f,B)$
\begin{enumerate}
\item
Generate a SysNF approximation $B'$ of $B$ with parameter $\eps / (\sqrt{n} \det(B))$: $\langle B' , \sigma,T \rangle$, and denote $L' = L(B')$,
and $N = \det(L')$.
\item
Generate
$$
\ket{\psi_1} = \sum_{x\in \Z^n} {\cal F}(T x/N) \ket{x}.
$$
\item
Add ancilla and apply $\Phi_3$ from Claim \ref{cl:dual}:
$$
\ket{x} \otimes \ket{0} \to  \ket{x + \Phi_3(x)} \otimes \ket{\Phi_3(x)}.
$$
Denote by $\ket{\psi_2}$.
\item
Apply the nearest-plane algorithm w.r.t. $({NL'})^*$ on the first register, and XOR the result into the second register:
$$
\ket{x} \otimes \ket{y} \to \ket{x} \ket{y + {\cal M}(x,N{B'}^{-T})}
$$
Denote by $\ket{\psi_3}$.
\item
Apply ${\cal F}_{L',N}$ on the first register:
$$
\ket{\psi_4} = ({\cal F}_{L',N} \otimes I) \ket{\psi_3}.
$$
\item
Measure the first register in $\ket{\psi_4}$: $\to z\in L'$.
Return $\sigma^{-1}(z)\in L$.
\end{enumerate}
\end{algorithm}
\end{mdframed}

\noindent
\\
We claim that using this algorithm we can PAC sample from any "nice" enough function:
\begin{theorem}\label{thm:sample}
Let $L = L(B)$ be some lattice.
Let $f$ be a function whose DFT, denoted by ${\cal F}$, is QES, and is also $(\nu,t)$ bounded
for $t \leq \lambda_1(L^*) / 2^{n/2+2}$.
Then for any $\eps = 2^{-\poly(n)}$ one can PAC-sample $|f|^2_{L}$ efficiently quantumly with parameters $(\eps,4\nu)$.
\end{theorem}

\subsection{Discussion}
It may be insightful already at this point to compare the above scheme to the
natural quantum scheme of generating super-positions
on a lattice that has already appeared in the work of Regev \cite{Reg09}:
in that scheme one generates a super-position on some discretization of real space,
that corresponds to the $L^*$-periodic FT
of the desired distribution, then "decodes" to the dual lattice $L^*$ using a CVP oracle for certain parameters,
and then applies QFT to achieve the desired super-position on the primal lattice $L$.

The main novelty of the proposed scheme is that it is encapsulated by the
SysNF reduction, in both ways.
This then requires the use of our newly defined quantum DFT for such lattices.
The reduction to SysNF allows to generate the initial coherent super-position on the trivial lattice $\Z^n$,
instead of some fine-grained version of the lattice itself as in \cite{Reg09}.
Notably, in order to generate a super-position on a fine-grained lattice $L/K$ for large $K$, one still requires access
to the input basis of $L$, which may make certain distributions inaccessible already at this initial step.

Let us compare the performance of our proposed algorithm with that of the classical algorithm of Gentry et al. \cite{GPV08}
for sampling from the discrete Gaussian distribution:
$$
\forall x\in L \quad \P(x) \propto e^{-\pi \|x\|^2 / s^2}.
$$
This classical algorithm samples from the Gaussian distribution with any variance $s^2$,  $s\geq \|\tilde{B}\| \sqrt{n}$
\footnote{The theorem statement in that paper is $\|\tilde{B}\| \omega(\sqrt{\ln n})$, but only for
a quasi-poly error in statistical distance, for exponentially small statistical distance
a bound of $\|\tilde{B}\| \sqrt{n}$ is required.}
where $\|\tilde{B}\|$ is the length of the longest vector when applying the Gram-Schmidt process to $B$.
Applying the LLL algorithm implies that $\|\tilde{B}\| \leq 2^{n/2} \lambda_1(L)$, and so 
that algorithm requires that $s$ be at least:
$$
s \geq 2^{n/2} \sqrt{n} \cdot \lambda_1(L) .
$$
To compare to our case, we use a result by Banaszczyk \cite{Ban93} showing that the discrete Gaussian $e^{-\pi \|x\|^2 / s^2}$ is 
$(s \sqrt{n}, 2^{-n})$ bounded for any lattice $L$.
In our algorithm we require that $t \leq \lambda_1(L^*) / 2^{n/2+2}$, hence
the lower bound on the standard-deviation of the output Gaussian $(1/t)$ is at most
$$
2^{n/2+2} / \lambda_1(L^*) \leq 2^{n/2+2} \lambda_1(L).
$$
Using $t = s \sqrt{n}$ implies that in our case the minimal lower-bound on the STD $s$ is given by
$$
s \geq 2^{n/2+2} \sqrt{n} \lambda_1(L).
$$
Thus, our lower-bound on $s$ is asymptotically the same as that of Gentry et al. \cite{GPV08}, for the discrete Gaussian distribution.

\subsection{Proof}
\begin{proof}
By Lemma \ref{lem:SysNF} the reduction generates a tuple $\langle B',\sigma,T\rangle$
such that for all $x\in L$ $\sigma(x)\in L(B')$ and
$$
\| \sigma(x)/T - x\| \leq \eps \|x\|,
$$
and $T = O(\det(B)/\eps) = 2^{\poly(n)}$.
Hence $L(B')$ can be represented efficiently by $\poly(n)$ bits.
Since $t \leq \lambda_1(L^*) / 2^{n/2+2}$ then
\be\label{eq:t1}
N t / T \leq \lambda_1(N{L}/T)^* / 2^{n/2+2} 
\leq 
\lambda_1(N{L'})^* / 2^{n/2+1}
\ee
where the last inequality follows from the prescribed SysNF approximation parameter $\eps$,
for all $\eps$ sufficiently small.
Since ${\cal F}$ is QES the state $\ket{\psi_1}$ can be generated efficiently.
By definition of $\Phi_3(x)$ we have:
$$
\ket{\psi_2} = \sum_{y\in N {L'}^*} \sum_{x\in L_N'} {\cal F}(T(x - y(x))/N) \ket{x} \ket{y(x)},
$$
where $y(x) := \Phi_3(x)\in NL^*$.
Consider the application of the nearest-plane algorithm.
Since 
${\cal F}(x)$ is $(\nu,t)$-bounded then
${\cal F}(x T/N)$ is $(\nu,N t/T)$ bounded, so we can approximate $\ket{\psi_2}$ to $l_2$ error at most $\nu$
by the following function:
$$
\ket{\psi_2} = \sum_{y\in N {L'}^*} \sum_{x\in L_N', \|x - y(x)\| \leq N t/T} {\cal F}(T(x - y(x))/N) \ket{x} \ket{y(x)},
$$
By Equation \ref{eq:t1} and Proposition \ref{prop:np1} for every $x$ in the support of $\ket{\psi_2}$ 
the call to the nearest-plane algorithm 
${\cal M}(x,N{B'}^{-T})$
returns a vector $y$ at distance at most 
$$
\|x - y(x)\| \cdot 2^{n/2} 
\leq
N t/ T  \cdot 2^{n/2}
\leq 
\lambda_1(N{L'})^*/2,
$$
away from the
closest $NL^*$ point from $x$ -- so it must be the correct point.
Hence, up to square-$l_2$-error at most $\nu$ we have:
$$
\ket{\psi_3} 
= 
\sum_{y\in N{L'}^*} \sum_{x\in L_N', \|x - y\| \leq N t/T} {\cal F}(T(x - y(x))/N) \ket{x} \otimes \ket{0}.
$$
We disregard the second register from now on, for clarity.
%By periodicity of $NL^*$ in $N$ we can rewrite the above as
%$$
%\ket{\psi_3} 
%= 
%\sum_{y\in N{L'}^*} \sum_{x\in L_N', \|x - y\| \leq N t/T} {\cal F}(T(x - y(x))/N) \ket{x} \otimes \ket{0}.
%$$
Using again the $(\nu,N t / T)$-bounded condition we approximate to $\nu$ error:
$$
\ket{\psi_3} 
= 
\sum_{y\in N{L'}^*} \sum_{x\in L_N'} {\cal F}(T(x - y)/N) \ket{x}.
$$
By Proposition \ref{prop:ft1} we get:
$$
\ket{\psi_4} = 
{\cal F}_{L',N} \ket{\psi_3}
=
\sum_{x\in L_N'} f(x/T) \ket{x}.
$$
Measuring this state yields $x\in L'$ according to the distribution 
$$
\P(x) = \frac{ |f(x/T)|^2}{\sum_{x\in L_N'} |f(x/T)|^2}.
$$
Applying the inverse SysNF map $\sigma^{-1}(x)$ yields by Lemma \ref{lem:SysNF} a vector $z\in L$
with probability:
$$
\P(z) \propto |f(z + \eps(z))|^2,
$$
where for each $z\in L_N'$ we have $\|\eps(z)\| \leq \eps \|z\| / (\sqrt{n} \det(B))$ 
by the prescribed error tolerance of the SysNF reduction.
Since $L$ is periodic in $\det(B) \cdot e_i$ for each $i\in [n]$ 
then we can assume w.l.o.g. that $\|z\| \leq \sqrt{n} \det(B)$, and so by our choice of parameters
$$
\|\eps(z)\| \leq \frac{\eps}{\sqrt{n} \det(B)} \sqrt{n} \det(B) = \eps.
$$
Hence $\P(z)$ is a PAC approximation of $f^2$ with parameters at most $(\eps,4\nu)$.

\end{proof}

\section*{Acknowledgements}
The authors thank Dorit Aharonov and Oded Regev for useful discussions and comments, and to Kevin Thompson for useful comments regarding this manuscript.
LE is supported by NSF grant no. CCF-1629809.
PS is supported by the NSF STC on Science of Information under Grant CCF0-939370, and by the NSF through grant CCF-121-8176.

\newcommand{\etalchar}[1]{$^{#1}$}

%\bibliographystyle{alphaurl}
%\bibliography{References}

\end{document}